\documentclass[a4paper,american,cleveref, autoref, thm-restate]{lipics-v2021}
\nolinenumbers 
\pdfoutput=1 
\hideLIPIcs  

\usepackage{amsmath,amssymb,amsfonts,amsthm}
\usepackage{mathtools}
\usepackage{dsfont}
\usepackage{xcolor}
\usepackage[textsize=footnotesize]{todonotes}
\usepackage{enumerate}
\usepackage[normalem]{ulem}
\usepackage{booktabs}

\newcommand{\R}{\mathds{R}}

\newcommand{\N}{\mathds{N}}

\newcommand{\Rst}{R_{s,t}^y}

\newcommand{\abs}[1]{\left\lvert{#1}\right\rvert}
\newcommand{\card}[1]{\left\lvert{#1}\right\rvert}

\newcommand{\floor}[1]{\left\lfloor{#1}\right\rfloor}

\newcommand{\define}{\coloneqq}
\newcommand{\set}[1]{\left\{#1\right\}}
\newcommand{\suchthat}{\,:\,}

\DeclareMathOperator{\sign}{sign}

\DeclareMathOperator*{\argmin}{argmin}

\DeclareMathOperator{\supp}{supp}

\title{Approximating the Network Design Problem for Potential-Based Flows} 

\titlerunning{Network Design for Potential-Based Flows}

\author{Max Klimm}{TU Berlin, Germany}{klimm@math.tu-berlin.de}{https://orcid.org/0000-0002-9061-2267}{}{}
\author{Marc E. Pfetsch}{TU Darmstadt, Germany}{pfetsch@mathematik.tu-darmstadt.de}{https://orcid.org/0000-0002-0947-7193}{}{}
\author{Martin Skutella}{TU Berlin, Germany}{skutella@math.tu-berlin.de}{https://orcid.org/0000-0002-9814-1703}{}{}
\author{Lea Strubberg}{TU Berlin, Germany}{strubberg@math.tu-berlin.de}{https://orcid.org/0009-0009-8505-3614}{}{}

\authorrunning{M.~Klimm, M.~Pfetsch, M.~Skutella, L.~Strubberg}
\Copyright{Max~Klimm, Marc~E.~Pfetsch, Martin~Skutella, Lea~Strubberg} 

\ccsdesc[500]{Mathematics of computing~Network flows}
\ccsdesc[500]{Mathematics of computing~Approximation algorithms}

\keywords{network design, approximation algorithms, potential-based flows} 

\relatedversion{} 

\funding{This work is funded by the DFG via the SFB TRR 154 subprojects A01 and A07.}

\begin{document}

\maketitle

\begin{abstract}
We develop efficient algorithms for a fundamental network design problem arising in potential-based flow models, which are central to many energy transport networks (e.g., hydrogen and electricity). In contrast to classical network flow problems, the nonlinearities inherent in potential-based networks introduce significant new challenges. We address these challenges through intricate reductions to classical combinatorial optimization problems, such as (constrained) shortest path problems, enabling the application of well-established algorithmic techniques to compute exact and approximate solutions efficiently. Finally, we complement these algorithmic results with matching complexity results concerning the hardness and non-approximability of the considered problem variants.
\end{abstract}

\section{Introduction}

Potential-based flow networks provide a versatile mathematical model for transport processes in infrastructure systems.
Formally, such a network is represented by a directed weakly connected graph $G=(V,A)$.
A flow is a vector $f\in\R^A$, where~$f_a>0$ denotes flow in the direction of arc~$a$ and~$f_a<0$ denotes flow in the opposite direction.
In contrast to classical flow models, feasible flows not only satisfy flow conservation, but are also induced by a node potential vector~$\pi\in\R^V$.
Specifically, flows and potentials are coupled through a \emph{potential-flow equation}
\begin{equation}
\label{eq:weymouth_con}
f_a = \mu_a \sign(\Delta \pi_a)\,\!\abs{\Delta \pi_a}^{1/r}
\qquad \text{where } \Delta\pi_a \define \pi_v-\pi_w \text{ for all } a=(v,w)\in A.
\end{equation}

Here, $\mu_a$ denotes the conductance of arc $a$.
Depending on the value of~$r > 0$, this framework captures several important application domains, including electrical, water, and stationary gas networks.
For $r=1$, it specializes to electrical networks: node potentials represent voltages, arc flows represent currents, and arc conductances are the reciprocals of Ohmic resistances.
For $r=2$, we recover a standard model of stationary gas transport in which node potentials represent squared pressures, arc flows represent volumetric gas flows, and conductances depend on physical pipe parameters such as roughness, length, and diameter.
Hence, potential-based flow networks are both practically relevant and mathematically challenging, since the discrete network structure is tightly coupled with nonlinear potential constraints.

In this paper, we study the design of minimum-cost potential-based flow networks that can route a prescribed amount of flow.
Such problems are particularly timely in view of the ongoing energy transition, which requires substantial investments in new infrastructure, for example in electricity and hydrogen networks.
As these projects require substantial capital expenditures, it is crucial to use mathematical optimization tools in their planning and design.

The input to this optimization problem consists of a directed graph~$G=(V,A)$, where each arc represents a candidate connection that may be installed.
Two nodes $s$, $t\in V$ are designated as source and sink, and the network must support an $s$--$t$ flow of value~$1$ using potentials in the interval~$[0,B]$ for a given bound~$B$.
Such bounds arise naturally in applications, for instance as safety or operational limits on voltages or pressures.
Our goal is to determine which arcs to install, and with which conductances, so that the resulting network can route one unit of flow at minimum cost within the given potential bounds.
Installing conductance $y_a$ on arc~$a$ incurs a fixed cost $\gamma_a\in\R_{\geq 0}$ and a variable cost $c_a\in\R_{\geq 0}$ per unit of conductance.
Thus, arc $a$ contributes no cost if $y_a=0$, and contributes $\gamma_a+c_a y_a$ if $y_a>0$.
In addition, each arc $a$ has an upper bound $\bar y_a\in\R_{\geq 0} \cup\{\infty\}$ on the conductance that can be installed.
The assumption that the flow is of value~$1$ can be made without loss of generality since~\eqref{eq:weymouth_con} implies that the flows scale linearly with conductances  so that other target flow values can be achieved by scaling the variable costs~$c_a$ appropriately.

For a node~$v$, let $\mathbf{1}_v$ denote the corresponding unit vector in $\R^V$.
Formally, the optimization problem considered in this paper is
\begin{subequations}
\label{eq:problem_intro}
\begin{align}
\inf_{x,y,\pi,f}\quad & \sum_{a \in A} \bigl(c_a y_a + \gamma_a x_a\bigr) \\
\text{s.t.}\quad
& \sum_{a \in \delta^+(v)} f_a - \sum_{a \in \delta^-(v)} f_a
    = \mathbf{1}_{s} - \mathbf{1}_{t}
    && \forall v \in V,  \\
& f_a = y_a \sign(\Delta \pi_a)\,\abs{\Delta \pi_a}^{1/r}
    && \forall a \in A,  \\
& y_a > 0 \;\Rightarrow\; x_a = 1 && \forall a \in A, \label{eq:problem_intro:ybound} \\
& \pi \in [0,B]^V,\quad
  f \in \R^A,\quad
  y_a \in [0,\bar{y}_a]~\forall a\in A,\quad
  x \in \{0,1\}^A, 
\end{align}
\end{subequations}
where $\delta^+(v)$ denotes the set of arcs leaving~$v$, and $\delta^-(v)$ denotes the set of arcs entering~$v$.
We note that, for every arc $a$ with $\bar{y}_a < \infty$, constraint~\eqref{eq:problem_intro:ybound} can be written as $0 \leq y_a \leq \bar{y}_a x_a$.
For all other arcs, the constraint can be formulated as $y_a( 1 - x_a) = 0$.
In any case, the problem can be formulated as a mixed-integer non-linear optimization problem (MINLP).

\subsection{Our Results}
We develop several exact and approximation algorithms for practically relevant variants of the optimization problem~\eqref{eq:problem_intro}. We complement these results with different hardness proofs showing that some simplifying assumptions are needed to obtain positive results. For an overview, see also \Cref{tab:results}.

We first observe that the considered network design problem~\eqref{eq:problem_intro} is $\mathsf{NP}$-hard and even hard to approximate. Our approximation hardness result uses a reduction from Steiner Tree and states that, even for the case without variable costs~($c\equiv0$),  there is no polynomial-time approximation scheme (PTAS), unless $\mathsf{P}=\mathsf{NP}$. In particular, the problem is strongly $\mathsf{NP}$-hard. We also show via reductions from the Partition Problem and the Knapsack Problem that the problem remains (weakly) $\mathsf{NP}$-hard on very simple series-parallel networks, even for the case without conductance bounds~($\bar{y}\equiv\infty$) and also for finite conductance bounds without variable costs~($c\equiv0$);
see Section~\ref{sec:complexity}.

In Section~\ref{sec:path-solutions} we study the important special case without conductance bounds~($\bar{y}\equiv\infty$) and prove that there is always an optimal solution consisting of a single $s$--$t$ path. If the variable costs or the fixed costs are zero for all arcs (i.e.,~$c\equiv0$ or~$\gamma\equiv0$), we show how to efficiently find such an optimal path by reducing the problem in a non-trivial way to a shortest path problem. For general costs, however, the problem is $\mathsf{NP}$-hard and we present a fully polynomial-time approximation scheme (FPTAS) that finds a near-optimal path solution. The latter algorithm is based on ideas related to the constrained shortest path problem. 

In Section~\ref{sec:series-parallel} we turn to the general case with finite conductance bounds ($\bar{y}\in\R_{\geq 0}^A$), which is hard to approximate, unless $\mathsf{P}=\mathsf{NP}$. For the special case of series-parallel networks, we present a fully polynomial-time approximation scheme that is based on a dynamic programming approach.
In contrast to standard approaches for related optimization problems on series-parallel networks (rounding etc.), we need to overcome the additional difficulty of how to choose the continuous variables~$y_a$, $a\in A$ (conductances).

\begin{table}[t]
\centering
\caption{%
Our results for the optimization problem~\eqref{eq:problem_intro} depending on whether conductance bounds are considered (``yes'': $\bar{y} \in (\R_{\geq 0} \cup \{\infty\})^A$; ``no'': $\bar{y} \equiv \infty$), fixed costs are allowed (``yes'': $\gamma \in \R_{\geq 0}^A$; ``no'': $\gamma \equiv 0$), variable costs are allowed (``yes'': $c \in \R_{\geq 0}^A$; ``no'': $c \equiv 0$), and whether arbitrary or series-parallel (SP) graphs are considered.
}
\label{tab:results}
\small
\begin{tabular}{@{}ccccll@{}}
\toprule
cond.~bd.~$\bar{y}$ & fix.~cost~$\gamma$ & var.~cost~$c$ & graph~$G$ & \multicolumn{2}{c}{results} \\
\midrule
yes       &   yes    & yes   & SP & FPTAS [Thm.~\ref{thm:SPFPTAS}] & $\mathsf{NP}$-hard~[Thm.~\ref{thm:NP_hardness}]  \\
yes       & yes  & no   & any & & $\mathsf{APX}$-hard [Thm.~\ref{thm:APXhard}] \\
yes & no  & yes & any & \multicolumn{2}{l}{conv.~opt.~[Cor.~\ref{cor:convex}]} \\
no        & no  & yes    & any   & polytime [Thm.~\ref{thm:poly_time}] \\
no        & yes & no        & any & polytime [Thm.~\ref{thm:poly_time}] \\
no        & yes & yes    & any   & FPTAS [Thm.~\ref{thm:poly_time}] & $\mathsf{NP}$-hard~[Thm.~\ref{thm:NP_hardness}] \\
\bottomrule
\end{tabular}
\end{table}

\subsection{Related Work}
The construction of cost-minimal networks supporting given flow demands is a long-standing research question in the optimization literature.
Most of the discussed models, however, consider the classical flow model where flow can be routed arbitrarily subject to capacity constraints on the arcs.
Chien~\cite{chien1960} considers one such model where for each pair of nodes a demand is given and the goal is to make a minimum investment in arc capacities such that for each pair of nodes the given demand can be routed.
Further improvements and variants of this model are discussed by Gomory and Hu~\cite{gomory1961,gomory1962}, Gusfield~\cite{gusfield1983}, Talluri~\cite{talluri1996}, Chou and Frank~\cite{chou1970}, Sridhar and Chandrasekaran~\cite{sridhar1992}, and Kabadi~\cite{kabadi2008}.

While the works above consider a set of $s$--$t$ flows that need to be realized in the network, Buchheim et al.~\cite{buchheim2011} consider the more general version in which the network needs to be able to support a set of different $b$-vectors and the network capacities are required to be integer.
They show that, for the case of at least three $b$-vectors, this problem is $\mathsf{NP}$-hard and provide a branch-and-cut framework for an arbitrary number of $b$-vectors.
Without the integrality constraints, the problem can be solved with linear programming techniques as shown by Schmidt~\cite{schmidt2014}.
For the more challenging problem where given network capacities need to be augmented to support a multi-commodity flow, Bienstock and Günlük~\cite{bienstock1996} and Günlük~\cite{Gunluk99} propose branch-and-cut frameworks.

Our problem formulation with fixed and variable costs for installing capacity is further reminiscent of the fixed-charge flow problem introduced by Hirsch and Dantzig~\cite{HischDantzig68}.
This problem generalizes the classical min-cost flow problem by allowing each arc to incur, in addition to the usual linear cost, a fixed cost that is charged whenever the flow on the arc is positive.
It contains the Steiner tree problem as a special case and is hence $\mathsf{NP}$-hard.
Valid inequalities and resulting branch-and-cut frameworks for this problem are studied, among others, by Ortega and Wolsey~\cite{ortega03}, Padberg, van Roy, and Wolsey~\cite{padberg85}, and van Roy and Wolsey~\cite{vanroy85}.

A network design for potential-based electrical networks (corresponding to the case $r=1$) has been studied by~Chan et al.~\cite{ChanLauSchildWongZhou2022}.
They considered an unweighted and purely discrete version of the problem where edges can be switched on or off; in the former case a given conductivity $\mu_a$ is used.
The authors ask for the minimal effective resistance that can be realized by switching $k$~edges on and show that this problem is $\mathsf{NP}$-hard.
It is easy to observe that this result immediately implies also that our problem \eqref{eq:problem_intro} is $\mathsf{NP}$-hard as well for the special case when $\gamma \equiv 1$, $c \equiv 0$, and $\bar{y} = \mu$.
Chan et al.~further showed that a weighted version of the problem is hard to approximate within a factor of $2-\epsilon$ under the small set expansion assumption.
Chan et al.~further provide an $8$-approximation for their problem under a technical assumption that allows them to buy up to $4$~copies of each edge while the optimum they compare to cannot buy each edge more than once.
The approximations we provide rely on different techniques and do not require such an assumption.
Mixed-integer nonlinear programming formulations and branch-and-bound frameworks have been proposed for various network design problems for potential-based flows~\cite{BorrazSanchezBentBackhausHijaziHentenryck2016,BragalliDAmbrosioLeeLodiToth12,BarrosAHCPT23,humpola15,HumS17,LiDeySahinidis24,raghunathan13}. Börner et al.~\cite{BornerKLPSS25} proposed a class of valid cuts that can be used in such frameworks.
In contrast to the algorithms we propose in this paper, these approaches are not guaranteed to run in polynomial time.

Bilevel optimization network design problems where the task is to support a traffic equilibrium (rather than a potential-based flow as in our paper) has been studied in the transportation literature \cite{abduul1979,GairingHarksKlimm2017,marcotte1986,marcotte1992}.
In this literature it is standard to assume that edges can be made arbitrarily efficient given sufficient investments. This corresponds to the case of no upper bounds on the conductance that we study in~\Cref{sec:path-solutions}.

Our approximation techniques in part solve \emph{constrained shortest path problems} as subproblems.
This problem is $\mathsf{NP}$-hard (Garey and Johnson~\cite{GareyJohnson1979}).
The first fully polynomial-time approximation scheme (FPTAS) for the problem on acyclic graphs was introduced by Hassin~\cite{hassin1992} and is based on dynamic programming, scaling, and rounding techniques.
His ideas were further improved  by Lorenz and Raz~\cite{LorenzRaz2001} who obtain a running time of $O(\card{A} n\,(\log \log n+1/\varepsilon))$ which also applies to general graphs.
Holzmüller~\cite{holzmueller2019} achieved an improved runtime of $O(\abs{A}n/\varepsilon)$ for undirected graphs.

\section{Preliminaries}
For a ground set $A$, a subset $S \subseteq A$, and an element $a \in A$, we denote by $\chi_S(a)$ the indicator function defined as $\chi_S(a) =1$ if $a \in S$, and $\chi_S(a) = 0$, otherwise.

In this section, we introduce the main notation of the paper.
In \Cref{subsec:potential-based-flow}, we first consider potential-based flow networks where each arc~$a$ has a given conductance~$\mu_a$.
In \Cref{subsec:network-design}, we introduce the problem of determining the conductances $y_a$ of each arc in a network in a cost-minimal way to satisfy certain constraints.
To distinguish between the cases where the conductances are given parameters and optimization variables, we denote them by~$\mu$ in the former case and by~$y_a$ in the latter case.

\subsection{Potential-Based Flow Networks}
\label{subsec:potential-based-flow}

Let $G=(V,A)$ be a weakly connected digraph with designated nodes $s\in V$ and~$t\in V$ with~$s\neq t$ which we call the \emph{source} and the \emph{sink} node of $G$, respectively.
The graph~$G$ can be represented by its node-arc incidence matrix~$\Gamma = (\gamma_{v,a})_{v \in V, a \in A}$ defined as~$\gamma_{v,a} = +1$ if  arc~$a$ leaves vertex~$v$, $\gamma_{v,a} = -1$ if arc~$a$ enters vertex~$v$, and $\gamma_{v,a} = 0$, otherwise.
A \emph{potential-based flow network} is a tuple~$(G,\mu)$ of a graph~$G$ and a conductance vector~$\mu \in \R_{>0}^A$. We are further given a parameter $r \geq 1$ which we assume to be fixed.

Given $(G, \mu)$ and $r$, a \emph{potential-based flow} in~$(G,\mu)$ is a tuple of a node-potential vector $\pi \in \R^V$ and an arc-flow vector~$f\in \R^A$ such that the potential-flow equation~\eqref{eq:weymouth_con} holds for every arc.
It is known, see, e.g., Birkhoff and Diaz~\cite{birkhoff1956}, that for every balance vector $b\in \R^V$ with $\sum_{v\in V} b_v=0$ there exists a unique potential-based flow~$(\pi,f)$ such that
\begin{subequations}
\label{potentialFlow}
\begin{align*}
    f_a = \mu_a \sign(\Delta \pi_a){\abs{\Delta \pi_a}}^{1/r}\quad\text{for all } a\in A,\qquad\Gamma f = b,\quad\text{and}\quad
    \pi_t = 0.
\end{align*}%
\end{subequations}%
In this case, we say that $(\pi,f)$ \emph{is induced by} $b$.
%
%
 We call the unique potential-based flow $(\pi,f)$ for $b=\mathbf{1}_s - \mathbf{1}_{t}$ a \emph{unit (potential-based) $s$--$t$ flow}.
The \emph{effective resistance} $R_{s,t}$ between~$s$ and~$t$ is defined as the difference between the potentials at nodes~$s$ and at~$t$ for the unit $s$--$t$ flow, i.e., 
\[R_{s,t}=\pi_s-\pi_t, \quad \text{where } (\pi,f) \text{ is a unit $s$--$t$ flow.} \]
Using that the flow along an arc~$a$ is nonzero only if $\Delta \pi_a \neq 0$ and that the potentials are monotone along a path that carries flow, we immediately obtain that~$R_{s,t}>0$.
We define the \emph{effective conductance} between $s$ and $t$ as $\smash{C_{s,t}= 1/{(R_{s,t})}^{1/r}}$.
Raber~\cite{Raber22} showed that the effective conductance is the value of the unique potential-based flow with potential difference of~$1$ between~$s$ and~$t$ and balance $0$ at every node but~$s$ and~$t$.
The following lemma establishes a simple upper bound on the effective resistance and effective conductance of a network that will be useful in the remainder of the paper.

\begin{restatable}{lemma}{boundsOnEffResAndEffCon} \label{lem:boundsEffResEffCon}
    Let $(G,\mu)$ be a potential-based flow network.
    \begin{enumerate}[(a)]
        \item For the effective resistance, we have
        $R_{s,t}\leq \sum_{a\in A}  1 / \mu_a^r$,
        where equality holds if $G$ is a path.
        \item For the effective conductance, we have
        $C_{s,t} \leq\sum_{a\in A} \mu_a$,
        where equality holds if all arcs in~$G$ are parallel from~$s$ to~$t$.
    \end{enumerate}
\end{restatable}

\begin{proof}
    Let $(\pi,f)$ be the unit potential-based flow from $s$ to $t$ in $G$ and let $P$ be an $s$--$t$ path in $G$. 
    We assume without loss of generality that the arcs are directed in the direction from $s$ to $t$ so that the flow $f_a$ and the potential difference $\Delta \pi_a$ of an arc $a\in P$ are non negative.
    Raber~\cite[Lemma~3.10]{Raber22} showed the intuitive result that no arc carries more flow than the flow value, i.e., that~$f_a\leq1$ for every arc~$a\in A$. Thus by the potential-flow equation~\eqref{eq:weymouth_con} it is
    \[ \mu_a {(\Delta \pi_a)}^{1/r} \leq 1 \quad \Leftrightarrow \quad \Delta \pi_a \leq  \frac{1}{\mu_a^r},\]
    where equality holds if and only if~$f_a=1$.
    Since $P$ is an $s$--$t$ path, the sum of the potential differences over every arc in $P$ equals the potential difference between node $s$ and $t$ so that
    \[R_{s,t} = \pi_s-\pi_t = \sum_{a\in P}\Delta \pi_a \leq \sum_{a\in P} \frac{1}{\mu_a^r} \leq \sum_{a\in A} \frac{1}{\mu_a^r}.\]
    If $G$ is a path, then clearly $f_a=1$ for all $a\in A$ and thus $R_{s,t}=\sum_{a\in A} 1 /\mu_a^r$.

    For the second inequality, let $S\subseteq V $ be an $s$--$t$ cut. 
    By the cut condition it is
    \[ 1 = \sum_{a\in \delta (S)} f_a \quad \Rightarrow \quad 1 \leq \sum_{a\in A} \abs{f_a}, \]
    where $\delta(S)$ is the set of arcs with one endpoint in $S$ and one endpoint not in $S$.
    Klimm et al.~\cite[Lemma~9]{Reduction23} showed that $0 \leq \pi_v \leq \pi_s$ for all nodes $v$. This implies that $\abs{\Delta \pi_a }\leq R_{s,t}$ for all arcs $a\in A$.
    Together with the potential-flow equation~\eqref{eq:weymouth_con}, we conclude
    \begin{equation} \label{eq:long_chain}
        1 \leq \sum_{a\in A} \abs{f_a} = \sum_{a\in A} \mu_a {\abs{\Delta \pi_a}}^{1/r} \leq \sum_{a\in A} \mu_a {(R_{s,t})}^{1/r} \quad \Rightarrow\quad \frac{1}{{{(R_{s,t})}^{1/r}}} \leq  \sum_{a\in A} \mu_a,
    \end{equation} 
    which implies the result by definition of $C_{s,t}$.
    
    If all arcs in $G$ are parallel, then $S=\set{s}$ is an $s$--$t$ cut and $\delta(S)=A$.
    Assuming without loss of generality that all arcs point from $s$ to $t$, it is $f_a>0$ and  $\Delta \pi_a =R_{s,t}$ for all $a\in A$.
    So in this case, all inequalities in~\eqref{eq:long_chain} hold with equality.
\end{proof}

\subsection{Network Design}
\label{subsec:network-design}

In this paper, we study the optimization problem~\eqref{eq:problem} of finding optimal investments in arc conductances to support a potential-based unit flow.
It will be helpful to rephrase the problem~\eqref{eq:problem} using the effective resistance.
To do so, let a vector $y \in \R_{\geq 0}^A$ of conductances be given. 
We define the set $\supp(y)\define \set{a\in A\suchthat y_a>0}$ of arcs with positive conductance as the \emph{support of $y$}. 
The \emph{network induced by $y$} is defined as~$G'=(V,A')$, with arc set~$A'=\supp(y)$ together with the projection of~$y$ to~$A'$. 
The effective resistance between $s$ and $t$ in the network induced by $y$ is the solution value to the following convex optimization problem (cf.~Klimm et al.~\cite{Reduction23}):
\begin{samepage}
\begin{equation}
\label{thomson} 
    \begin{aligned}
        R_{s,t}^y =  \min \Biggl\{ \sum_{a\in A'}\frac{1}{y_a^r}\abs{f_a}^{r+1} : f \in \R^{A'} \text{ with } \Gamma' f = \mathbf{1}_{s} - \mathbf{1}_{t} \Biggr\}, 
    \end{aligned}
\end{equation}
\end{samepage}%
where $\Gamma'$ denotes the node-arc incidence matrix of $G'$.
Note that the induced effective resistance~$R_{s,t}^y$ is infinite if and only if $s$ and $t$ are not connected in $G'$.
Analogously, we denote by $C_{s,t}^y$ the effective conductance between $s$ and $t$ in the network induced by $y$ so that,
\begin{equation*}
    C_{s,t}^y=  \begin{cases}
                            {(\Rst)}^{-1/r} & \text{if } R_{s,t}^y <\infty, \\
                            0                       & \text{otherwise}.
                \end{cases}
\end{equation*}

Using the fact that the highest potential of an $s$--$t$ flow is attained at node~$s$ (cf.~\cite[Lemma~9]{Reduction23}), we reformulate \eqref{eq:problem_intro} as follows:
\begin{subequations}
\label{eq:problem}
\begin{align}
\inf_{x,y}\quad & \sum_{a \in A} \bigl(c_a y_a + \gamma_a x_a\bigr) \\
\text{s.t.}\quad
& R^y_{s,t} \leq B \label{eq:problem:R_st} \\
& 0 \leq y_a \leq x_a \bar{y}_a && \text{for all~} a \in A, \label{eq:problem:bound} \\
& x \in \{0,1\}^A
\end{align}
\end{subequations}

Constraint~\eqref{eq:problem:bound} implies zero conductance~$y_a = 0$ for every arc~$a$ with $x_a=0$. On the other hand, an optimal solution to Problem~\eqref{eq:problem} with zero conductance on an arc can be modified by setting the binary variable of that same arc to zero without increasing the objective value of that solution.
For a feasible solution $(x,y)$ to Problem~\eqref{eq:problem} we can therefore assume without loss of generality that 
\[x_a =1 \quad \text{ if and only if} \quad y_a >0.\]

Note that this formulation is a bilevel optimization problem since Constraint~\eqref{eq:problem:R_st} requires the solution of the inner optimization problem~\eqref{thomson} to determine the effective resistance for a vector of conductances~$y$.
We first show that this optimization problem is monotone in the sense that the effective resistance decreases as $y$ increases. 

\begin{restatable}{lemma}{effResMonotone}
    \label{lem:effResMonotone}
    The effective resistance $\Rst$ is monotonically decreasing if any component~$y_a$, $a \in A$, increases.
\end{restatable}

\begin{proof}
    Let $\tilde{y}, \hat y\in \R^A_{\geq 0}$ be such that there is $a^*\in A$ with $\tilde{y}_a=\hat y_a$ for all $a\neq a^*$ and $\tilde{y}_{a^*}\leq \hat y_{a^*}$.
    Let $\tilde{G} = (V, \tilde{A})$ and $\hat{G} = (V, \hat{A})$ be the graph induced by $\tilde{y}$ and $\hat{y}$, respectively.
    If $R_{s,t}^{\tilde{y}} = \infty$, then clearly $\R_{s,t}^{\tilde{y}} \geq R_{s,t}^{\hat y}$.
    We may therefore assume that there exists an optimal solution $\tilde{f} \in \R^{\tilde{A}}$ to Problem~\eqref{thomson}.
    We define $\hat f_a \define \tilde{f}_a$ for all arcs $a \in \tilde{A}$ in the support of~$\tilde{y}$.
    If $\tilde{y}_{a^*} = 0$ and $ \hat y_{a^*} >0$, then arc $a^*$ is not in the support of $\tilde{y}$ but in the support of $\hat y$ and we set $\hat f_{a^*}\define 0$. Otherwise, $\hat A$ contains exactly the arcs in $\tilde{G}$.
    In both cases, $\hat f$ is an $s$--$t$ flow in~$\hat{G}$, because $\tilde{f}$ is an $s$--$t$ flow in~$\tilde{G}$. 
    The flow $\hat f$ is therefore feasible for Problem~\eqref{thomson} for $\hat y$ and $\hat G$.
    Hence,
    \[
    R_{s,t}^{\tilde{y}} = \sum_{a\in \tilde{A}}\frac{1}{\tilde{y}_a^r}\abs{\tilde{f}_a}^{r+1}=  \sum_{a\in \hat A} \frac{1}{\hat y_a^r} \abs{\hat{f}_a}^{r+1} \geq R_{s,t}^{\hat y}.\qedhere
    \]
\end{proof}

The effective resistance is also convex, a property that is already known for electrical networks; see, e.g., \cite{GhoshBoydSaberi2008} and \cite{ChanLauSchildWongZhou2022}. 

\begin{restatable}{lemma}{effResConvex}
    \label{lem:effResConvex}
    The effective resistance $\Rst$ is convex in $y$.
\end{restatable}

\begin{proof}
    We interpret the effective resistance~$\Rst$ as the optimal value function of the parametric nonlinear optimization problem~\eqref{thomson} and apply results on general parametric convex optimization.
    Let  $g\colon \R^A \times \R^A_{\geq 0}\rightarrow \R$ be a scalar function, $\smash{F\colon \R_{\geq 0}^A\rightarrow 2^{\R^A}}$ a point-to-set-map.
    Consider a general parametric optimization problem of the form
    \[ \min \{ g(f,y) \, \colon\,   f\in F(y)\}. \]
    Then, the optimal value function~$g^* \colon \R_{\geq 0}^A \rightarrow \R\cup \set{\infty}$ 
    is defined  as
    \begin{equation*}
        g^*(y) \define    \begin{cases}
                            \inf \{g(f,y) \suchthat f\in F(y)\} & \text{if } F(y) \neq \emptyset,\\
                            + \,\infty & \text{otherwise}.
                    \end{cases} 
    \end{equation*}
    Fiacco and Kyparisis~\cite[Proposition~2.1]{FiaccoKyparisis1986} showed that the optimal value function~$g^*$ is convex on a set $S\subseteq \R^A_{\geq 0}$ if the following three conditions are satisfied:
    \begin{enumerate}
        \item $S$ is convex,
        \item $g$ is convex on the set $\set{(f,y)\suchthat f\in F(y), \,y\in S}$,
        \item $F$ is convex on $S$,
    \end{enumerate}
    where the point-to-set map $F$ is convex on $S$ if for all $y$, $y'\in S$ and $\lambda \in (0,1)$, it holds that
    \[ \lambda F(y) + (1-\lambda) F(y') \subseteq F\big(\lambda\, y + (1-\lambda) y'\big).\]

    We consider $S\define \R_{\geq 0}^A$, which is convex.
    To prove the convexity of $\Rst$ on $S$, we use a slightly different problem formulation than the one stated in~\eqref{thomson}. 
    Here, we define $F(y)$ for $y\in S$ as the set 
    \begin{equation} \label{eq:Fyset}
        F(y)\define \set{f\in \R^{A}\suchthat \Gamma f =\mathbf{1}_s - \mathbf{1}_{t}, \, \supp(f)\subseteq \supp(y)}
    \end{equation} 
    where $\Gamma$ is the node-arc incidence matrix of $G$.
    For $f\in F(y)$ and $y\in S$, we define $g(f,y)$ as
    \begin{equation*}
        g(f,y)\define \sum_{a\in \supp(y)} \frac{1}{y_a^r}\abs{f_a}^{r+1}.
    \end{equation*}
    It can be verified that  $g^*(y)=\Rst$ for all $y\in S$.

    We show that all conditions of the Proposition are satisfied.
    The set $S$ is convex.
    The function $g$ is the sum of functions of the form
    \[
    g_a(f_a,y_a) =  \begin{cases}
                    \frac{1}{y_a^r}\abs{f_a}^{r+1} & \text{if } y_a > 0, \\
                    0 & \text{if } f_a = 0,\\
                    \infty & \text{otherwise},
                \end{cases} 
    \] 
    which are perspective functions of the convex function $h(f_a) = \abs{f_a}^{r+1}$.
    Thus, $g_a$ is convex on the set $\set{(f,y)\suchthat f\in F(y), \,y\in S}$ for each $a\in A$~\cite{rockafellor1970}.
    Therefore, $g$ is convex as well.
    
    For the convexity of $F$, let $y,y'\in S$, $f\in F(y)$, $f'\in F(y')$, and $\lambda\in (0,1)$.
    We show that $\lambda f +(1-\lambda)f'\in F(\lambda y+(1-\lambda) y')$.
    Due to the linearity of the matrix-vector multiplication, we have
    \[ \Gamma \big(\lambda f +(1-\lambda)f' \big) =  \lambda  \Gamma f +(1-\lambda) \Gamma f' = \lambda (\mathbf{1}_s - \mathbf{1}_{t}) + (1-\lambda)(\mathbf{1}_s - \mathbf{1}_{t}) = \mathbf{1}_s - \mathbf{1}_{t}.\]
    To prove that $\supp(\lambda f+(1-\lambda) f')\subseteq \supp (\lambda y+(1-\lambda) y')$, we show that
    \[ \lambda y_a+(1-\lambda) y'_a = 0 \quad \Rightarrow \quad \lambda f_a+(1-\lambda) f_a' = 0.\]
    So let $a\in A$ such that $\lambda y_a+(1-\lambda) y'_a = 0$.
    Then $y_a = 0$ and $y'_a=0$ because $y_a\geq 0$, $y'_a\geq 0$, and $\lambda >0$ by assumption.
    Because $f\in F(y)$, $\supp(f) \subseteq \supp(y)$ follows and  the same holds for $f'$ and $y'$.
    Therefore, since $y_a = 0$ and $y'_a=0$, it follows that $f_a = 0$ and $f'_a=0$. 
    This implies $\lambda f_a+(1-\lambda) f_a' = 0$.
    Thus, $\lambda f +(1-\lambda)f'\in F(\lambda y+(1-\lambda) y')$, and $F$ is therefore convex.
    
    All three conditions of the Proposition are satisfied so that $g^*$ is convex on $S$.
    We conclude that $\Rst$ is convex in $y$.
\end{proof}

As an immediate corollary, we obtain that the Network Design Problem~\eqref{eq:problem} can be solved with convex optimization techniques when there are no fixed costs.

\begin{corollary}
\label{cor:convex}
The Network Design Problem~\eqref{eq:problem} in the case $\gamma \equiv 0$ admits a convex formulation and can therefore be approximated to arbitrary precision by the ellipsoid method.
\end{corollary}

Finally, note that constraint~\eqref{eq:problem:R_st}  can be replaced by 
\begin{equation}
    \tag{4b$'$} \label{problem:effcon}
    C_{s,t}^y \geq D \quad \text{with } D \define \frac{1}{{B}^{1/r}},
\end{equation}
where $C_{s,t}^y $ is the effective conductance between $s$ and $t$ in the network induced by $y$.
In the technical sections, we will use both constraint variations, since we find it sometimes more convenient to work with the effective conductance instead of the effective resistance.

\section{Network Design without Conductance Bounds}
\label{sec:path-solutions}

In this section, we consider the Network Design Problem~\eqref{eq:problem} in the case where the bound on the conductance is infinite, i.e., $\bar{y}_a=\infty$ for all arcs $a \in A$.
As we will show in~\Cref{thm:NP_hardness} (Section~\ref{sec:complexity}), this special case is still $\mathsf{NP}$-hard.
In this section, we show that the problem admits a fully polynomial-time approximation scheme (FPTAS) and is even solvable in polynomial time if we have only fixed costs or no fixed costs at all, i.e., when $c \equiv 0$ or when $\gamma \equiv 0$.

\begin{theorem}\label{thm:poly_time}
We consider Problem~\eqref{eq:problem} without conductance bounds, i.e., $\bar{y} \equiv \infty$.
\begin{enumerate}[(a)]
\item\label{thm:poly_time:a} If $c \equiv 0$ or $\gamma \equiv 0$, the problem can be solved in polynomial time.
\item\label{thm:poly_time:b} For arbitrary costs~$c$, $\gamma\in\R_{\geq0}^A$, there is a fully polynomial-time approximation scheme.
\end{enumerate}
\end{theorem}

In \Cref{subsec:path} we show that, without conductance bounds, Problem~\eqref{eq:problem} always has an optimal solution whose support is an $s$--$t$ path. In \Cref{subsec:shortest-path} we reduce the problem with special cost structure described in \Cref{thm:poly_time}~\eqref{thm:poly_time:a} to a shortest path problem which can be efficiently solved. Finally, in \Cref{subsec:constrained-shortest-path}, we develop a fully polynomial-time approximation scheme for the case of arbitrary costs.


\subsection{Existence of Optimal Path Solutions}
\label{subsec:path}

For the case without conductance bounds ($\bar{y} \equiv \infty$), we show that Problem~\eqref{eq:problem} always has an optimal solution whose support is an $s$--$t$ path. 

\begin{lemma} 
When~$\bar{y} \equiv \infty$, then there exists an $s$--$t$ path $P$ and an optimal solution $(x,y)$ to Problem~\eqref{eq:problem} such that $x_a = 0$ and $y_a = 0$ for all $a \notin P$.
\end{lemma}

\begin{proof}
    We start by proving the result for the special case where the fixed costs~$\gamma \equiv 0$.
    Let $(x,y)$ be an optimal solution to Problem~\eqref{eq:problem} and let $(f,\pi)$ be the potential-based unit $s$--$t$ flow in the network induced by~$y$, so that $\Rst =\pi_s-\pi_t$.
    We may assume without loss of generality that the arcs are oriented such that the flow values~$f_a$ and the potential differences~$\Delta \pi_a = \pi_u-\pi_v$ are non-negative for all arcs~$a=(u,v)\in A$. Since positive flow only occurs on arcs with positive potential difference,~$f$ is acyclic and can be decomposed into flow along $s$--$t$ paths~$\mathcal P$. That is,~$f_P>0$ units of flow are sent along path~$P\in\mathcal P$ such that~$f_a=\sum_{P\in{\mathcal P}\colon a\in P}f_P$ and~$\sum_{P\in\mathcal P}f_P=1$.
    
    We may even assume that the paths in $\mathcal P$ are pairwise arc-disjoint:
    Otherwise, let $\mathcal P_a\subseteq \mathcal P$ be the set of paths that share a fixed arc $a\in A$.
    Split arc~$a$ into $\abs{\mathcal P_a}$ many parallel copies denoted by $a_P$ for every $P\in \mathcal P_a$.
    When sending a flow of value $f_P$ through every arc $a_P$, it is clear that we obtain a disjoint flow-decomposition when doing this for every arc in the original solution. 
    If we set the conductance to $y_{a_P}\define f_P/{(\Delta \pi_a)}^{1/r}$, then the potential vector~$\pi$ together with the split flow clearly fulfills the potential-flow equation in~\eqref{eq:weymouth_con}.
    In particular, the effective resistance does not change when an arc $a$ is split like this.
    Also the cost of the solution remains the same when setting the variable cost of every arc $a_P$ to $c_a$, because 
    \[
    \sum_{P\in\mathcal P_a}c_a\, y_{a_P}= \frac{c_a}{(\Delta \pi_a)^{1/r}} \sum_{P\in\mathcal P_a} f_P= \frac{c_af_a}{(\Delta \pi_a)^{1/r}} = c_ay_a,
    \] 
    due to the definition of $y_{a_P}$, the flow-decomposition, and the potential-flow equation~\eqref{eq:weymouth_con}.
    We may therefore assume from now on that the paths in $\mathcal P$ are pairwise arc-disjoint.
    
    Each path $P\in\mathcal P$ induces a partial cost $c_P\define \sum_{a\in P}c_a y_a$.
    So, the objective value of the considered feasible solution is equal to $\sum_{P\in\mathcal P}c_P$.
    We show that a path with the lowest ratio of partial cost and flow value yields an optimal solution.
    Let  $P^*\in \mathcal P$ be such a path, i.e.,
    \begin{equation} \label{eq:Pstar}
        P^*\in \argmin_{P\in\mathcal P}{\frac{c_P}{f_P}},
    \end{equation}
    and define $(x^*,y^*)$ by $x^*_a \define 1$ and $y^*_a \define y_a/f_{P^*}$ for all $a\in P^*$, and~$x^*_a\define y^*_a\define 0$ otherwise.
    The effective resistance of a path is the sum of the resistances of its arcs due to \Cref{lem:boundsEffResEffCon}. So,
    \begin{equation*}
        R_{s,t}^{y^*} = \sum_{a\in P^*} \frac{1}{(y_a^*)^r} = \sum_{a\in P^*} \frac{(f_{P^*})^r}{(y_a)^r} = \sum_{a\in P^*} \frac{(y_a)^r \Delta\pi_a}{(y_a)^r} = \pi_s-\pi_t = R_{s,t}^y \leq B,
    \end{equation*}
    where we use the definition of $y_a^*$, the potential-flow equation~\eqref{eq:weymouth_con}, and the telescoping sum $\sum_{a\in P^*} \Delta \pi_a = \pi_s-\pi_t$ along $s$--$t$ path~$P^*$. 
    The solution~$(x^*,y^*)$ is therefore feasible and its cost is
    \begin{equation*}
        \sum_{a\in P^*}y_a^* c_a = \sum_{a\in P^*} \frac{y_a}{f_{P^*}} c_a =   \frac{1}{f_{P^*}} c_{P^*}  = \sum_{P\in \mathcal P} f_P\frac{c_{P^*}}{f_P^*} \leq \sum_{P\in \mathcal P} f_P\frac{c_{P}}{f_P} = \sum_{P\in\mathcal P} c_P,
    \end{equation*}
    using the definition of $y_a^*$, the definition of $c_{P^*}$, the fact that $1 = \sum_{P\in\mathcal P}{f_P}$, and the definition of $P^*$ in~\eqref{eq:Pstar}.
    The cost of $(x^*,y^*)$ is therefore bounded by the cost of $(x,y)$ and thus optimal. This concludes the proof for the special case of fixed costs~$\gamma\equiv0$.

    It remains to extend the result to arbitrary fixed costs. Consider an optimal solution~$(x^*,y^*)$ and let~$G^*=(V,A^*)$ be the subgraph given by the arc set~$A^* \define \{a\in A \colon x^*_a=1\}$. Notice that~$y^*$ minimizes~$\sum_{a\in A^*}c_ay^*_a$ subject to the constraints~$R^{y^*}_{s,t}\leq B$ and~$\supp(y^*)\subseteq A^*$. In other words, the restriction of~$(x^*,y^*)$ to subgraph~$G^*$ is an optimal solution to Problem~\eqref{eq:problem} on~$G^*$ with~$\gamma\equiv0$. As we argued above, there is such an optimal solution that sends flow along a single $s$--$t$ path. Since the path is contained in $A^*$, setting $x_a =1$ only for the arcs on the path may only decrease the fixed costs.
\end{proof}

Let $\mathcal P$ denote the set of all (arc sets of) $s$--$t$ paths in $G$. As there always exists an optimal solution to Problem~\eqref{eq:problem} such that $\supp(x)$ is an $s$--$t$ path, we reformulate the problem as follows.
We restrict the decision variable vector $x$ to only those arc vectors which induce an $s$--$t$ path~$P\in\mathcal P$.
Since, by \Cref{lem:boundsEffResEffCon}, the effective resistance of a path is the sum of arc resistances along this path, we can replace constraint~\eqref{eq:problem:R_st} by the following inequality:
\[\sum_{a\in P} 1/y_a^{r} \leq B.\]
Furthermore, for every solution, decreasing the conductance of an arc on the path such that the inequality is satisfied with equality does not increase the cost.
Thus, Problem~\eqref{eq:problem} with $\bar{y} \equiv \infty$ is equivalent to the following problem:
\begin{samepage}%
\begin{subequations} 
\label{path_problem}
\begin{alignat}{4}
    \underset{P,y}{\inf}        & \quad&&\sum_{a\in P} (c_a \,y_a +\gamma_a)  &\quad &  \label{path_problem:objective}\\
    \text{s.t.} &      && \sum_{a\in P} 1/y_a^{r} = B,  &\quad & \label{path_problem:weymouth_pressure}  \\
                &      && y\in \R^P_{>0}, & \quad &  \\
                &      && P\in  \mathcal P.&  \quad &
\end{alignat}
\end{subequations}
\end{samepage}%

Notice that for a given solution~$(P,y)$ to Problem~\eqref{path_problem}, the conductance~$y_a$ of any arc~$a\in P$ with zero cost~$c_a=0$ can be chosen arbitrarily large without changing the cost of the solution. In this case~$1/y_a^r$ tends to zero and we can safely ignore this arc in constraint~\eqref{path_problem:weymouth_pressure}. In what follows, we denote by $P_+=P\cap \supp(c)$ the set of arcs on path $P\in\mathcal P$ with positive variable cost and always assume that~$y_a$ is chosen infinitely large for any arc~$a\in P\setminus P_+$.


\subsection[Proof of Theorem 5(a)]{Proof of \Cref{thm:poly_time}~\eqref{thm:poly_time:a}}
\label{subsec:shortest-path}

Our proof of \Cref{thm:poly_time}~\eqref{thm:poly_time:a} is based on the following lemma.

\begin{lemma}\label{lem:KKT}
    For every $s$--$t$ path $P\in\mathcal P$ there exists a conductance vector $y\in \R_{>0}^{P}$ such that $(P,y)$ is feasible for Problem~\eqref{path_problem} and such that there is no other feasible conductance vector on $P$ with strictly lower cost value.
    Furthermore, there exists $\lambda>0$ such that
    \begin{equation} \label{eq:KKT}
    \lambda =\frac{c_a \,{y}^{r+1}_a}{r}  \quad \text{for all } a\in P_+.
    \end{equation}  
\end{lemma}
\begin{proof}
    Consider Problem~\eqref{path_problem} and fix the path~$P$ in an optimal solution $(P,y)$. Substituting~$z_a=1/y_a^{r}$ for all $a \in P_+$, the optimal conductance $y$ can be determined from the solution to the following convex program:
    \begin{samepage}
    \begin{alignat*}{3}
        \inf        & \quad&&\sum_{a\in P_+} c_a/z_a^{1/r} \\
        \text{s.t.} &      && \sum_{a\in P_+} z_a = B, \\
                    &      && z\in \R^{P_+}_{>0}. 
    \end{alignat*}
    \end{samepage}%
    Then, Slater's condition holds for any $z\in \R_{>0}^{P_+}$  which satisfies
    $z_a = B/\card{P_+}$ for all $a\in P_+$,
    where $\card{P_+}$ denotes the number of arcs contained in $P_+$.
    Therefore, any $z\in \R_{>0}^{P_+}$ and $\lambda  \in \R$ for which strong duality holds must fulfill the KKT conditions.
    Hence,
    \begin{equation}
        \lambda = \frac{c_a}{rz_a^{\frac{r+1}{r}}} \quad \text{for all~} a\in P_+.
    \end{equation}
    Substituting $z$ back to $y$ yields the result.
\end{proof}

With \Cref{lem:KKT} we reformulate Problem~\eqref{path_problem} even further.
Using~\eqref{eq:KKT}, we replace the variables $y_a$ in Constraint~\eqref{path_problem:weymouth_pressure}, and obtain
\begin{align*}
\sum_{a\in P_+} {(r\lambda/c_a)}^{-\frac{r}{r+1}} &= B & &\Leftrightarrow & 
\lambda &= \frac{\Big(\sum_{a\in P_+}{c_a^{\frac{r}{r+1}}}\Big)^{\frac{r+1}{r}}}{rB^{\frac{r+1}{r}}}.
\end{align*}
With this we eliminate variable $\lambda$ in~\eqref{eq:KKT} and get
\begin{equation}\label{eq:y_explicit}
    y_a=\frac{\Big(\sum_{a'\in P}{c_{a'}^{\frac{r}{r+1}}}\Big)^{\frac{1}{r}}}{c_a^{\frac{1}{r+1}}B^{\frac{1}{r}}} \quad \text{for } a \in P_+.
\end{equation}
Now, we eliminate the variable $y$ in the objective, so that
\[
\sum_{a\in P} (c_ay_a+\gamma_a) = \Bigg(\sum_{a\in P} {\biggl(\frac{c_a}{{B}^{1/r}}\biggr)}^{\frac{r}{r+1}} \Bigg)^{\frac{r+1}{r}} +\sum_{a\in P}\gamma_a.
\]
Let $c'\in \R_{\geq0}^A$ be defined such that $c'_a = {(c_a/B^{1/r})} ^{r/(r+1)}$ for all $a\in A$.
Then,  Problem~\eqref{path_problem} is equivalent to the
following optimization problem
\begin{samepage}
\begin{align}\label{quadratic_path_problem}
    \underset{P\in\mathcal P}{\min} \Big(\sum_{a\in P} c'_a \Big)^{\frac{r+1}{r}} +\sum_{a\in P}\gamma_a.  
\end{align}
\end{samepage}

With this formulation we finally prove~\Cref{thm:poly_time}~\eqref{thm:poly_time:a}.
If $c\equiv0$, Problem~\eqref{quadratic_path_problem} is simply a shortest path problem with arc lengths $\gamma$.
In the other case, where $\gamma\equiv0$, define the arc lengths as $c_a^{{r}/{(r+1)}}$ for all $a\in A$.
Then, the shortest path $P$ with respect to this length together with vector $y$ defined as in~\eqref{eq:y_explicit} is optimal for Problem~\eqref{eq:problem} because
\[
    \argmin_{P\in \mathcal P} \Bigg(\sum_{a\in P}{(c_a/{B}^{\frac{1}{r}})}^{\frac{r}{r+1}} \Bigg)^{\frac{r+1}{r}} =
     \argmin_{P\in \mathcal P} \Bigg(\sum_{a\in P}{c_a}^{\frac{r}{r+1}} \Bigg)^{\frac{r+1}{r}}  =
     \argmin_{P\in \mathcal P} \, \sum_{a\in P}{c_a}^{\frac{r}{r+1}}
\]
since $g(x)={x}^{(r+1)/r}$ is monotonically increasing.

\begin{remark}
\label{remark:irrational}
Notice that potential-based flows frequently involve irrational numbers like, e.g., $c_a^{r/(r+1)}$ above for given rational or even integral input~$c_a$. Since we talk about efficient algorithms in \Cref{thm:poly_time}, we need to be somewhat more careful. For the particular result in \Cref{thm:poly_time}~\eqref{thm:poly_time:a}, it is enough to assume that instead of the conductance costs~$c_a$, we are given arc lengths~$c_a^{r/(r+1)}$ as rational numbers for all arcs~$a\in A$. Notice that our efficient shortest path algorithm then only needs to find a shortest $s$--$t$ path for the given arc lengths.

These issues also play a role in subsequent sections. However, as we only present algorithms for finding approximate solutions there, we can more easily overcome these issues by working with bounded precision in all our computations. To simplify the presentation, in what follows we only work with exact calculations (involving irrational numbers) without going into the rather tedious and less interesting details of bounded precision.
\end{remark}

\subsection[Proof of Theorem 5 (b)]{Proof of \Cref{thm:poly_time}~\eqref{thm:poly_time:b}}
\label{subsec:constrained-shortest-path}

We now turn to the general case where both cost components~$c$ and~$\gamma$ are non-zero.
Building on the reformulation in~\eqref{path_problem}, we derive a fully polynomial-time approximation scheme (FPTAS). The key idea is to approximate the optimal dual variable and find for this approximation a shortest path with bounded effective resistance. In light of \Cref{thm:poly_time}~\eqref{thm:poly_time:a}, we may assume  for the following arguments that $c \not\equiv 0$.

Let $(P^*,y^*)$ be an optimal solution to Problem~\eqref{path_problem}.
We write $P^*_+ = P^* \cap \supp(c)$ for the arcs on the optimal path with non-zero variable costs.
For every $\varepsilon \in (0,1/3]$ we want to find a path $P'\in \mathcal P$ and feasible conductances $y' \in \R^{P'}_{>0}$ such that
\[
\sum_{a\in P'} c_a \,y'_a +\gamma_a \leq (1+\varepsilon) \sum_{a\in P^*} \bigl(c_a \,y^*_a +\gamma_a\bigr).
\]
The algorithm is based on approximating the optimal dual variable~$\lambda^*$.
This dual variable exists for $P^*$ due to \Cref{lem:KKT} such that the following holds:
\begin{equation} \label{eq:KKT_star}
\lambda^* =\frac{c_a \,{(y^*_a) }^{r+1}}{r}  \quad \text{for all~} a\in P^*_+.
\end{equation}
To construct the FPTAS, we proceed as follows. 
First, we derive bounds on the optimal dual variable~$\lambda^*$ to restrict the search space.
Then, for every fixed $\lambda$ within these bounds we define arc conductances and find a path with bounded effective resistance that minimizes the overall cost.
We discretize the range of $\lambda$ to efficiently find a good approximation to $\lambda^*$.
Finally, we prove that the approximation is correct and analyze the run time.

We start with the  bounds on $\lambda^*$.
Due to the equality constraint~\eqref{path_problem:weymouth_pressure} there cannot exist an arc $a'$ on~$P^*$ such that $1/{(y_{a'}^*)}^r > B$.
Therefore, $y_a^*\geq B^{-1/r}$ for all arcs $a\in P^*$.
Together with Equation~\eqref{eq:KKT_star}, this yields the following lower bound:
\[ L \define \min_{a\in A : c_a \neq 0} \frac{ c_a}{r \,\!{B}^{\frac{r+1}{r}}} \leq  \lambda^*.\]
Furthermore, there is an arc $a'\in P^*_+$ with $1/{(y^*_{a'})}^r \geq  B /{\card{P^*_+}}$. Otherwise, $\sum_{a\in P^*_+} 1/{(y_a^*)}^{r} < B$.
Therefore, $y_{a'}^*\leq \card{P^*_+}^{1/r} \big/ B^{1/r}$.
Together with Equation~\eqref{eq:KKT_star}, we obtain the following upper bound on $\lambda^*$:
\[ U \define \max_{a\in A : c_a \neq 0} \; {c_a} \frac{(n-1)^{{\frac{r+1}{r}}}}{r \, B ^{{\frac{r+1}{r}}}}  \geq \frac{c_{a'}\,\card{P^*}^{{\frac{r+1}{r}}}}{r \, B^{{\frac{r+1}{r}}}} \geq  \lambda^*.\]

For a fixed $\lambda$ with $L\leq \lambda \leq U$, we define a corresponding conductance vector $y^\lambda \in \R^A_{\geq 0}$ on the arc set $A_+ \coloneqq \{a \in A : c_a \neq 0\}$ according to~\eqref{eq:KKT} as 
\begin{equation}
    \label{eq:y_lambda}
    y_a^\lambda \define
    \left(\frac{\lambda r}{c_a}\right)^{\frac{1}{r+1}} \quad \text{for } a\in A_+.
\end{equation}
The conductances of the other arcs $\{a \in A : c_a = 0\}$ are implicitly set to $\infty$ meaning that they approach $\infty$ in the infimum.

For fixed conductance $y^\lambda$, let $P^\lambda$ be a cost-minimal $s$--$t$ path such that $(P^\lambda,y^\lambda)$ is a feasible solution to Problem~\eqref{path_problem}.
Then path~$P^\lambda$  solves the following optimization problem:
\[\min_{P\in\mathcal P} \Bigg\{ \sum_{a\in P_+} (c_a y^\lambda_a +\gamma_a)  + \sum_{a \in P \setminus P_+} \gamma_a \, \colon \, \sum_{a\in P_+} 1/{(y_a^\lambda )}^r \leq B \Bigg\}. \]
Eliminating the $y^\lambda_a$ variables by~\eqref{eq:y_lambda}, we reformulate this problem in terms of~$\lambda$:
\begin{equation}%
\label{restr_path_problem}
    g(\lambda) \coloneqq \underset{P\in \mathcal P}{\min} 
    \left\{\sum_{a\in P} \bigl({(\lambda r)}^{\frac{1}{r+1}}\,{c_a}^{\frac{r}{r+1}}  +\gamma_a\bigr)~\colon~
    \sum_{a\in P} {c_a}^{\frac{r}{r+1}} \leq {(\lambda r)}^{\frac{r}{r+1}} B\right\} .
\end{equation}%
Note that in this formulation arcs~$a$ with $c_a = 0$ are treated correctly: regardless of $\lambda$ they contribute only their fixed cost~$\gamma_a$ to the objective, and they do not contribute to the effective resistance bound.

Observe that $P^*$ is an optimal solution to Problem~\eqref{restr_path_problem} for $\lambda = \lambda^*$ and that $\lambda^*$ minimizes~$g$ over $[L,U]$, i.e., $g(\lambda^*) = \min\{g(\lambda) \,\colon \, \lambda \in [L,U]\}$.
The optimal dual variable $\lambda^*$ and corresponding optimal path $P^*$ can therefore be found by minimizing $g$ over $[L,U]$ and solving the optimization problem~\eqref{restr_path_problem} for $\lambda^*$. 
However, as Problem~\eqref{path_problem} is $\mathsf{NP}$-hard, we do not expect that this can be done in polynomial time.

Furthermore, it is $\mathsf{NP}$-hard to compute $g(\lambda)$ because Problem~\eqref{restr_path_problem} defines an instance of the restricted shortest path problem~\cite{GareyJohnson1979} as follows.
For each arc $a\in A$ we are given a cost, in our case $(\lambda r)^{1/(r+1)}c_a^{r/(r+1)}  +\gamma_a$, and a length that in our case corresponds to $c_a^{r/(r+1)}$.
A restricted shortest path is a cheapest $s$--$t$~path in terms of cost, such that the total length does not exceed a given bound, in our case $(\lambda r)^{r/(r+1)}B$.
Several fully polynomial-time approximation schemes, the first one introduced by Hassin~\cite{hassin1992}, were developed for the restricted shortest path problem.
So instead we discretize the search space $[L,U]$ and compute only an approximation of $g$ at these values.
We now describe how this can be done so that the overall minimum yields a good approximation to the minimum of $g$.
Let $\varepsilon \in(0,1]$. 
We discretize the range of $\lambda$ by defining
\[
\lambda_i \define \bigl(1+ \tfrac{\varepsilon}{3}\bigr)^{i(r+1)}L  \quad \text{and} \quad I\define \{i\in \N_{0} \,\colon \, L\leq \lambda_i\leq U\}.
\]
Using an FPTAS for the restricted shortest path problem, we compute a feasible path $P^i \in \mathcal P$ that approximates the optimal solution of Problem~\eqref{restr_path_problem} for $\lambda = \lambda_i$, i.e.,
\begin{equation*} 
    \sum_{a\in P^i} \Bigl({(\lambda_i r)}^{\frac{1}{r+1}}\,{c_a}^{\frac{r}{r+1}} +\gamma_a\Bigr) \leq \bigl(1+ \tfrac{\varepsilon}{3}\bigr) g(\lambda_i).
\end{equation*}
We show now that the discretization of the search space is sufficient to find a good approximation for $\lambda^*$ and $P^*$.
Let $j\in I$ be such that $ \lambda_{j-1}\leq \lambda^*\leq\lambda_j$.
Then, $P^*$ is a feasible solution to Problem~\eqref{restr_path_problem} with $\lambda =\lambda_j$ since $\lambda^*\leq \lambda_j$.
Furthermore, by choice of~$j$ we get 
\[\lambda_j\leq \bigl(1+ \tfrac{\varepsilon}{3}\bigr)^{r+1}\lambda_{j-1} \leq \bigl(1+ \tfrac{\varepsilon}{3} \bigr)^{r+1}\lambda^*.\]
Then $P^j$ is an $(1+\varepsilon)$-approximation to $P^*$ because
\begin{align*} 
    \sum_{a\in P^j} &\bigl(y_a^{\lambda_j}c_a+\gamma_a\bigr) =
    \sum_{a\in P^j} \bigl({(\lambda_j r)}^{\frac{1}{r+1}}\,{c_a}^{\frac{r}{r+1}} +\gamma_a\bigr) 
    \leq \bigl(1+ \tfrac{\varepsilon}{3}\bigr) g(\lambda_j)\\
    &\leq \bigl(1+ \tfrac{\varepsilon}{3}\bigr) \sum_{a\in P^*} \bigl({(\lambda_j r)}^{\frac{1}{r+1}}\,{c_a}^{\frac{r}{r+1}}  +\gamma_a \bigr) 
    \leq \bigl(1+ \tfrac{\varepsilon}{3}\bigr)^{\!2} \sum_{a\in P^*} \bigl({(\lambda^* r)}^{\frac{1}{r+1}}\,{c_a}^{\frac{r}{r+1}}  +\gamma_a \bigr) \\
    & \leq (1+{\varepsilon}) \sum_{a\in P^*} \bigl({(\lambda^* r)}^{\frac{1}{r+1}}\,{c_a}^{\frac{r}{r+1}}  +\gamma_a \bigr)
    = (1+{\varepsilon}) \sum_{a\in P^*} \bigl(y_a^*c_a +\gamma_a\bigr),
\end{align*}
using the definition of $y^{\lambda_j}$ in~\eqref{eq:y_lambda}, that $P^j$ approximates the optimal value $g(\lambda_j)$ of Problem~\eqref{restr_path_problem} for $\lambda=\lambda_j$, that $P^*$ is feasible to this problem, that $\lambda_j$ approximates $\lambda^*$, and, for the last inequality, the fact that $(1+\varepsilon/3)^2 \leq 1+\varepsilon$ for $\varepsilon\leq 1/3$.

Now, consider the minimum cost value of $P^i$ over $I$ and let
\[
i'\in \argmin \Biggl\{\sum_{a\in P^i} \bigl(y_a^{\lambda_i}c_a+\gamma_a\bigr) \,\colon\, i\in I\Biggr\}
\]
be the index of this minimum.
Define $y'\define y^{\lambda_{i'}}$ and  $P'\define P^{i'}$.
Then $(P',y')$ is a $(1+\varepsilon)$-approximation to Problem~\eqref{path_problem}, because
\begin{align*}
    \sum_{a\in P'} \bigl(y_a'c_a+\gamma_a\bigr) \leq \sum_{a\in P^j} \bigl(y_a^{\lambda_j}c_a+\gamma_a\bigr) \leq  (1+{\varepsilon}) \sum_{a\in P^*} \bigl(y_a^*c_a +\gamma_a \bigr).
\end{align*}

We finally turn to the analysis of the running time.
For every $i\in I$ we run an FPTAS for a restricted shortest path instance.
The ratio of the maximal and the minimal arc cost is
\[C\define \frac{\max_{a\in A : c_a \neq 0 }c_a}{ \min_{a\in A : c_a \neq 0}c_a}.\] 
Then $ \card{I} \in O\bigl({\log (n  C)/(\varepsilon r})\bigr)$ since 
\[
\lambda_i = L ( 1+ \varepsilon/3)^{i(r+1)} \leq U \quad \Leftrightarrow \quad  i(r+1) \leq \log_{1+\varepsilon/3 }(U/L) 
\]
and 
\[\log_{1+\varepsilon/3 }(U/L) \leq \frac{\log_{2 }(U/L)}{\log_2(1+\varepsilon /3)} \leq \frac{3 \log_2\bigl((n-1)^{\frac{r+1}{r}} C\bigr)}{\varepsilon},\]
using the fact that $x\leq \log_2(1+x)$ for $x\in [0,1]$. 
Holzmüller shows in~\cite{holzmueller2019} that the run time of Hassin's FPTAS can be improved to $O(\card{A}n/\varepsilon)$ for undirected graphs. 
We can apply his result as the direction of the arcs in $G$ is not relevant in Problem~\eqref{path_problem}.
This yields an overall running time of 
$O\bigl(\abs{A}n (\,\log n +\log C) /({\varepsilon^2}r)\bigr)$.\qedhere

\section{Series-Parallel Network Design with Conductance Bounds}
\label{sec:series-parallel}

The Network Design Problem with conductance bounds ($\bar{y} \in \R_{\geq 0}^A$) is hard to approximate on general graphs, see \Cref{thm:APXhard}.
The problem remains $\mathsf{NP}$-hard on series-parallel graphs as we show in~\Cref{thm:NP_hardness}.
It admits, however, a fully polynomial-time approximation scheme on this graph class if the variable costs are strictly positive. 

In preparation for this result, we first study a variant of the Network Design Problem, where the conductance is fixed. This problem variant has also been considered in~\cite{BornerKLPSS25} for general potential-based flows and in~\cite{ChanLauSchildWongZhou2022} for electrical networks.

Formally, the input consists of a graph $G$ with designated nodes $s$ and $t$, a positive conductance vector $\mu \in \R_{>0}^A$, and a cost vector $p \in \N^A$. 
The goal is to find a cost-minimal subgraph of $G$ such that the effective resistance between $s$ and $t$ is not greater than a given  bound $B>0$.
The problem has the following MINLP formulation:
\begin{samepage}
\begin{subequations}
\label{IPCO_st}
\begin{alignat}{3}
    \underset{x}{\min}        & \quad&&\sum_{a\in A} p_a \,x_a  &\quad & \label{IPCO_st:objective} \\
    \text{s.t.} &     && R^{x,\mu}_{s,t} \leq B,  &\quad & \label{IPCO_st:effRes} \\
                &      && x\in \set{0,1}^A, \label{IPCO_st:binary} & \quad &
\end{alignat}
\end{subequations}
\end{samepage}%
where $R^{x,\mu}_{s,t}$ denotes the effective resistance between $s$ and $t$ in the graph induced by $x$ with arc conductance  $\mu_a$ for every arc $a\in \supp(x)$. 

This problem is a special case of Problem~\eqref{eq:problem} with conductance bounds and vanishing variable cost. 
Formally, we reduce an instance of Problem~\eqref{eq:problem} with $c\equiv0$ to Problem~\eqref{IPCO_st} by setting $\mu_a$ to $\bar{y}_a$ and the cost vector $p$ to the fixed cost vector $\gamma$.
The reduction is valid since any optimal solution $(x,y)$ to Problem~\eqref{eq:problem} with $c\equiv0$ can be modified by increasing the conductance $y_a>0$ of an arc $a$ with $x_a=1$ to $\bar{y}_a$.
This decreases the effective resistance due to \Cref{lem:effResMonotone} so that the modified solution remains feasible and  optimal as the variable cost is zero.

It is straightforward to show that Problem~\eqref{IPCO_st} is $\mathsf{NP}$-hard, even on series-parallel graphs.

\begin{theorem} \label{thm:NP_hardness_bounded}
    Problem \eqref{IPCO_st} is $\mathsf{NP}$-hard even on series-parallel graphs. 
\end{theorem}

\begin{proof}
    Consider the graph consisting of $m$ parallel arcs $a_1,\dots, a_m$ each connecting the only two nodes $s$ and $t$. 
    If each arc $a_i$ has cost $p_i$ and conductance $\mu_i$, Problem~\eqref{IPCO_st} yields the following formulation:
    \begin{samepage}
    \begin{alignat*}{3}
        \underset{x}{\min}        & \quad&&\sum_{i=1}^m p_i \,x_i  &\quad &  \\
        \text{s.t.} &     && \sum_{i=1}^m \mu_i \,x_i \geq D,  &\quad &  \\
                    &      && x\in \set{0,1}^m, & \quad &
    \end{alignat*}
    \end{samepage}%
    since the effective conductance of the graph induced by $x$ is given by $C_{s,t}^x= \sum_{i=1}^m \mu_ix_i$ due to the parallel composition of the arcs.
    The problem is also known as Knapsack cover or min-Knapsack problem which is $\mathsf{NP}$-hard, like the classical Knapsack problem.
\end{proof}

We begin with discussing an exact algorithm for Problem~\eqref{IPCO_st} that exploits the structure of the effective resistance of series-parallel graphs.
Furthermore, we explain how this algorithm can be adapted to an FPTAS using classic scaling and rounding techniques.
Finally, we show how these results generalize to a fully polynomial-time approximation scheme for the Network Design Problem~\eqref{eq:problem} with conductance bounds on series-parallel graphs with general cost.

\subsection{An Exact Algorithm}\label{subsec:exact_fixed}

An exact dynamic programming algorithm for Problem~\eqref{IPCO_st} on series-parallel graphs was first proposed by Chan et al.~\cite{ChanLauSchildWongZhou2022} for electrical networks, i.e., the case where $r=1$.
In the following, we extend this algorithm for general potential-based flow networks.

The algorithm exploits the recursive structure to compute the effective resistance of series-parallel graph compositions. 
This composition rule is well known for electrical networks and has been extended to potential-based flows by Gross et al.~\cite{Gross19}.
\begin{lemma}[\cite{Gross19}] \label{lem:effResSPGraphs}
    For $i\in\set{1,2}$, let $(G^{(i)},\mu^{(i)})$ be a potential-based flow network with source node $s^{(i)}$ and sink node $t^{(i)}$.
    Let $R_{i}$ and $C_i$ be the effective resistance and the effective conductance in $G^{(i)}$ between $s^{(i)}$ and $t^{(i)}$, respectively
    \begin{enumerate}
        \item If $G$ is a series-composition of $G^{(1)}$ and $G^{(2)}$, that we obtain by identifying node $t_1$ with $s_2$, then the effective resistance $R_{s,t}$ and effective conductance $C_{s,t}$ between node $s=s_1$ and $t=t_2$ in $G$ are given by
        \[ R_{s,t} = R_1 +R_2 \quad \text{and} \quad C_{s,t} = \frac{C_1\,C_2}{{\Big({C}^r_1+{C}^r_2\Big)}^{1/r}}.\]
        \item If $G$ is a parallel-composition of $G^{(1)}$ and $G^{(2)}$, that we obtain by identifying node $s_1$ with $s_2$ as well as node $t_1$ with $t_2$, then the effective resistance $R_{s,t}$ and effective conductance $C_{s,t}$  between node $s=s_1$ and $t=t_1$ in $G$ are given by
        \[ R_{s,t} = \frac{R_1\,R_2}{{\Big({{R}_1^{1/r}}+{{R}_2^{1/r}}\Big)}^{r}} \quad \text{and} \quad C_{s,t} = C_1+ C_2.\]
    \end{enumerate}
\end{lemma}

For the algorithm, we use a tree representation of a series-parallel graph $G$, referred to as an \emph{SP-tree} of $G$, see, e.g.,~\cite{BrandtstaedtVanBangSpinrad1999}.
An SP-tree~$T$ is a rooted full binary tree in which each node~$v$ represents a connected subgraph $G_v$ of $G$, with the root corresponding to the entire graph~$G$.
Every inner node $v$ of the tree corresponds to either a parallel or a series composition of the two subgraphs represented by the two children of $v$, while the leaves of $T$ represent the single arcs in $G$.
Since $T$ is a full binary tree with $m\define \card{A}$ leaves, it has exactly $2m-1$ nodes.
Every series-parallel graph admits such an SP-tree, which can be computed in linear time using the algorithm proposed by Valdes et al.~\cite{ValdesTarjanLawler1982}. 

The effective resistance of the graph $G$ can be computed by processing the SP-tree from bottom to top according to \Cref{lem:effResSPGraphs}. 
We denote the effective resistance between~$s_v$ and~$t_v$ in $G_v$ by~$R_v$ and abbreviate series and parallel composition by \emph{SC} and \emph{PC}, respectively.
The effective resistance of $G$ is then computed recursively by the following rules:
\begin{align*}
    R_{v} &= \frac{1}{\mu_a^r } && \text{if $v$ is a leaf representing arc $a\in A$,} \\
    R_{v} &= {R_u}+{R_w} && \text{if $G_v$ is an SC of $G_u$ and $G_w$, }\\
    R_{v} &= \frac{R_u\,R_w}{{\Big({{R}_u^{1/r}}+{{R}_w^{1/r}}\Big)}^{r}} && \text{if $G_v$ is a PC of $G_u$ and $G_w$. }   
\end{align*}

Let $U$ be an upper bound on the optimal solution to Problem~\eqref{IPCO_st}, e.g., $U =\sum_{a\in A} p_a$.
For the dynamic program, we define for every node $v$ in the SP-tree of $G$ and every possible cost $k\in\{0,1,\dots,U\}$ the following subproblem
\begin{align}
    R(v,k) = \min_{H\subseteq G_v}\left\{ R^{H}_{s_v,t_v} \,\colon \, \sum_{a\in H}p_a \leq k\right\},
\end{align}
which computes the minimum effective resistance between $s_v$ and $t_v$ over every subgraph of $G_v$ with arc cost of at most $k$.

Let $v_0$ denote the root node of the SP-tree.
The optimal value of Problem~\eqref{IPCO_st} is then given by 
\[ \min_{k\in \{1,\dots,U\}}\{k\,\colon\, R(v_0,k)\leq B\}. \]
With the following recurrence, we compute $R(v,k)$ from the leaves of the tree to the root:
\begin{align*} R(v,k) = \begin{cases}
    \infty & \text{if $v$ is a leaf and $k<p_a$},\\
    R_{v} & \text{if $v$ is a leaf and $k\geq p_a$},\\
    \underset{k'\in \{0, \dots,k\}}{\min} R(u,k') + R(w,k-k') & \text{if $G_v$ is an SC of $G_u$ and $G_w$},\\
    \underset{k'\in \{0, \dots,k\}}{\min} \frac{R(u,k')\,R(w,k-k')}{{\Big({{\big(R(u,k')\big)}^{1/r}}+{{\big(R(w,k-k')\big)}^{1/r}}\Big)}^{r}} & \text{if $G_v$ is a PC of $G_u$ and $G_w$}.
\end{cases}
\end{align*}
Here are $O(mU)$ subproblems, since the SP-tree has $O(m)$ nodes and we consider $U+1$ values of $k$ for each node.
Solving each subproblem takes $O(U)$ time since we enumerate all $k'\in \{0,\dots,k\}$, so that the overall time complexity of the dynamic program is $O(mU^2)$.

\subsection{FPTAS for the Network Design Problem with Discrete Conductance} \label{subsec:FPTAS_fixed}

Based on the dynamic program, we obtain a fully polynomial-time approximation scheme for arbitrary $p \in \R_{\geq 0}^A$ by a classic scaling and rounding ansatz.
Therefore, let $\varepsilon\in (0,1]$.
For a feasible solution $x$, let $p(x)= \sum_{a\in A}p_a x_a$.

Let $x^*$ be an optimal solution, and let $p_{\max} \define \max\{p_a : a\in \supp{(x^*)}\}$. 
Then $p_{\max}$ is a lower bound on the optimal value, while $U=m\, p_{\max}$ is an upper bound, i.e., $p_{\max} \leq p(x^*) \leq m \,p_{\max}$.

To obtain a polynomial running time, we scale the cost as follows 
\[\rho_a\define \floor{\frac{p_a}{\delta}} \, \forall a\in A \quad \text{with} \quad \delta \define \frac{\varepsilon \, p_{\max{}}}{m},\]
so that $p_a-\delta \, \rho_a \leq \delta$ and  $\delta \, \rho_a \leq p_ a$ for all arcs $a\in A$.
Let $x'$ be the solution returned by the dynamic program for $G$ with scaled cost vector $\rho$.
As $x^*$ is a feasible solution to the problem with scaled cost, it is $\rho(x')\leq \rho(x^*)$.

Then, $x'$ is a $(1+\varepsilon)$-approximation to Problem~\eqref{IPCO_st} because
\begin{align*}
    p(x')   & = p(x')-\delta \, \rho(x') + \delta \, \rho(x') \\
            & \leq m\, \delta + \delta \, \rho(x^*) \\
            & \leq \varepsilon \, p_{\max{}} + \,p(x^*) \\
            & \leq (1+\varepsilon) p(x^*).
\end{align*}

With the scaled cost and upper bound $U'=m \floor{{p_{\max}}/{\delta}} \leq  m^2/  \varepsilon$ , the dynamic program runs in time $O(m(U')^2)= O(m^5/{\varepsilon ^2})$.
However, $p_{\max}$ is not known a priori.
We therefore enumerate all possible values $p_{\max}= p_a$  and run the  dynamic program with scaled cost  for every $a\in A$.
Finally, we compare the objective value of every enumerated solution and return the minimum.
This yields in an overall running time of $O(m^6/\varepsilon ^2)$.

\subsection{FPTAS for the General Network Design Problem}
The exact algorithm and the resulting FPTAS described in \Cref{subsec:exact_fixed} and \Cref{subsec:FPTAS_fixed} can be easily adapted to the case where there is not only one choice for the conductance of each arc, but a number of conductance values if this number is polynomially bounded in the input size.

We now extend this FPTAS for the discrete problem to the general Problem~\eqref{eq:problem} in the case $c\in \R^A_{>0}$ with conductance bounds and variable conductance on series-parallel graphs.
The key idea is to discretize the conductance interval $[0,\bar{y}_a]$ for every arc $a\in A$ into sufficiently many values, thereby obtaining an instance of Problem~\eqref{IPCO_st} with multiple conductance options.

\begin{restatable}{theorem}{SPFPTAS} \label{thm:SPFPTAS}
    The Network Design Problem~\eqref{eq:problem} with $\bar{y} \in \mathbb{R}_{\geq 0}$ and $c\in \mathbb R^A_{>0}$ admits a fully polynomial-time approximation scheme for series-parallel graphs.
\end{restatable}

\begin{proof}
    The proof proceeds as follows.
    First, we derive a lower bound on the optimal value of Problem~\eqref{eq:problem}.
    Based on this bound, we construct a lower bound on the conductance of each arc, ensuring that every selected arc contributes a non-negligible fraction of the optimal value.
    To obtain an instance of Problem~\eqref{IPCO_st} with multiple conductance options, we discretize the feasible conductance interval of each arc into polynomially many values.
    We show that the discretization introduces only a small multiplicative error.
    More precisely, we prove that the optimum of the discretized instance is a $(1+\varepsilon/3)$-approximation to Problem~\eqref{eq:problem}. 
    Applying the FPTAS described in \Cref{subsec:FPTAS_fixed} with parameter $\varepsilon/3$ to the discretized instance yields a desired $(1+\varepsilon)$-approximation to Problem~\eqref{eq:problem}.
    
    Let $(G,c,\gamma,\bar{y},D)$ be an instance of Problem~\eqref{eq:problem} using constraint formulation~\eqref{problem:effcon}, where $G=(V,A)$ is a series-parallel graph with source  $s\in V$ and sink  $t\in V$, and let $m\define\abs{A}$. 

    We first derive a lower bound on the optimal objective value.
    Let $(x^*,y^*)$ be an optimal solution. 
    Using the constraint~\eqref{problem:effcon} and \Cref{lem:boundsEffResEffCon}, we obtain
    \[ D\leq C_{s,t}^{y^*} \leq \sum_{a\in A} y^*_a. \]
    Hence, there exist an arc $a^*\in A$ such that $\frac{D}{m} \leq y^*_{a^*}$.
    This implies the follower lower bound on the objective: 
    \[L\define \min_{a\in A} c_{a}\frac{D}{m}+\gamma_{a} \leq \sum_{a\in \supp(x^*)} c_ay^*_a+\gamma _a. \]

    Let $\varepsilon\in (0,1)$ be given.
    For each arc $a\in A$, we define a lower bound on the conductance by
    \[
        \underline{y}_a \define \frac{\varepsilon \, L}{6 \,c_a m}
    \]
    This choice ensures that the cost of an arc with $c_a\neq 0$ is at least $c_a\underline{y}_a= \varepsilon L/(6m)$. 

    We now discretize the conductance interval for each arc $a\in A$.
    For $i\in \N_0$,  define 
    \[
        \mu_{a_i} \define \bigl(1+ \tfrac{\varepsilon}{6}\bigr)^{\! i} \underline{y}_a \quad \text{and} \quad I_a \define \set{i\in \N_0 \suchthat \underline{y}_a \leq \mu_{a_i} \leq \bar{y}_a} .
    \]
    In addition, we include the value $\mu_{a_{-1}}\define \bar{y}_a$ for every arc~$a$. 
    The number of discrete conductance values $\abs{I_a}$ is polynomially bounded in the input size because
    \[ \mu_{a_i} = \bigl(1 + \tfrac{\varepsilon}{6}\bigr)^{\! i}\underline{y}_a \leq \bar{y}_a \quad \Leftrightarrow \quad i \leq \log_{1+\varepsilon/6}(\bar{y}_a/\underline{y}_a) \leq \frac{6 }{\varepsilon} \log_2\left(\frac{\bar{y}_a 6\,c_a m}{\varepsilon L}\right).\]

    We construct an instance of Problem~\eqref{IPCO_st} by assigning for each arc~$a\in A$ the conductance option $\mu_{a_i}$ with corresponding cost 
    \[
        p_{a_i}\define c_a\, \mu_{a_i} +\gamma_a \quad \text{for all }i\in I_a\cup\set{-1}.
    \]
    This yields an instance of polynomial size.
    Let $(x',\mu')$ be an optimal solution to this discretized instance.
    We compare this solution with the optimal solution $(x^*, y^*)$ of the original problem. 
    For each arc $a \in \supp(x^*)$, let $i^*$ be the smallest index such that $y^*_a \le \mu_a^{i^*}$. If no such index exists, we set $i^* = -1$.
    Then the solution obtained from $(x^*, y^*)$ by replacing $y^*_a$ with $\mu_a^{i^*}$ is feasible for the discretized problem, since increasing the conductance can only decrease the effective resistance; see \Cref{lem:effResMonotone}.

    We now bound the increase in the objective value. Let
    \[
        S\define \set{ a\in \supp(x^*) \suchthat y_a^*\leq \underline{y}_a}
    \]
    be the set of optimal arcs with small conductance value.
    Then, for $a\notin S$, we have 
    \[
         \mu_{a_{i^*}} - y^*_a \leq  \frac{\varepsilon}{6} y^*_a
    \]
    while for $a\in S$, we have $\mu_{a_{i^*}}= \underline{y}_a$. 
    Therefore,
    \begin{align*}
        \sum_{a\in \supp(x')} p_{a_{i'}} - \sum_{a\in \supp(x^*)} c_ay^*_a+\gamma _a & \leq \sum_{a\in \supp(x^*)} p_{a_{i^*}} - \sum_{a\in \supp(x^*)} c_ay^*_a+\gamma _a \\
        & = \sum_{a\in \supp(x^*)} c_a(\mu_{a_{i^*}}-y^*_a) \\
        & = \sum_{a\in S} c_a(\underline{y}_a-y^*_a) + \sum_{a\notin \ S} c_a(\mu_{a_{i^*}}-y^*_a) \\
        & \leq \sum_{a\in S} c_a\underline{y}_a + \frac{\varepsilon}{6}\sum_{a\notin S} c_ay^*_a \\
        & = \abs{S} \frac{\varepsilon \,L}{6\, m} + \frac{\varepsilon}{6}\sum_{a\notin S} c_ay^*_a \\
        & \leq   \frac{\varepsilon }{6} L + \frac{\varepsilon}{6}\sum_{a\notin  S} c_ay^*_a \\
        & \leq \frac{\varepsilon }{3} \sum_{a\in \supp(x^*)} c_ay^*_a+\gamma _a,
    \end{align*}
    where we use the definition of $\underline{y}_a=\frac{\varepsilon \, L}{6\,c_am}$ and the fact that $L$ and $\sum_{a\in \supp(x^*)\setminus S} c_ay^*_a$ are lower bounds on the optimal value.

    Thus, the discretized solution $(x',\mu')$ is a $(1+\varepsilon/3)$-approximation of the original problem.
    Finally, we apply the FPTAS from \Cref{subsec:FPTAS_fixed} with accuracy $\varepsilon/3$ to the discretized instance.
    This introduces an additional factor of $(1+\varepsilon/3)$, resulting in an overall approximation factor of $(1+\varepsilon)$.
\end{proof}

\section{Complexity Results}
\label{sec:complexity}

In this final section, we discuss the computational complexity of the general Network Design Problem and its variants.
We first provide an approximation-preserving reduction from the Steiner tree problem which implies (using the non-approximability result of Chlebík and Chlebíková~\cite{ChlebikC08}) that the Network Design Problem is $\mathsf{NP}$-hard to approximate with a factor of $96/95$ and, in particular, does not admit a PTAS. 

\begin{restatable}{theorem}{APXhard} \label{thm:APXhard}
    For $r\geq 1$ there is no polynomial-time approximation scheme for the Network Design Problem~\eqref{eq:problem}, unless $\mathsf{P} = \mathsf{NP}$, even in case $c\equiv 0$ and $\bar{y}\in \R_{\geq 0}^A$.
\end{restatable}

\begin{proof} 
    We present an approximation-preserving reduction from the Steiner Tree Problem.
    Given an instance of the Steiner Tree Problem, that is, a graph~$G=(V,E)$ with terminal nodes~$\{t_0,t_1,\dots,t_m\}\subseteq V$ and arc costs~$\gamma \in \N_{0}^E$, we construct an instance of the Network Design Problem as follows.
    
     Let $G'=(V\cup \{t\}, A\cup \{(t_i,t)\,\colon \, i\in [m] \})$ be the graph obtained from $G$ by connecting an additional node $t$ to all terminal nodes except the terminal $s\define t_0$ and choosing an arbitrary orientation~$A$ of the edges in~$E$.
    We assign the edges incident to $t$ a fixed cost of zero, while we retain the fixed cost of the other edges from the Steiner tree instance.
    The variable cost vector $c$ is set to the zero vector.
    For the upper bounds of the conductance, we give a large bound for the arcs in $G$, namely $\bar{ y}_a =n \, m $ for $a\in A$, where $n=\abs{V}$, and a small bound for the new arcs, that is $\bar{y}_a =1/m$ for $a$ incident to~$t$.
    We set the bound on the effective resistance to $B =1+\frac{n-1}{n^r m^r}$.

    We claim that there exists a solution to the Steiner Tree Problem if and only if there exists a feasible solution of the same cost to the Network Design instance constructed above. 
    To show this, let $H\subseteq G$ be a feasible solution to the Steiner Tree Problem, i.e., $H$ is a connected subgraph of $G$ containing all terminal nodes. 
    To construct a feasible solution~$(x,y)$ to the Network Design Problem, we set~$y_a=n\,m$ and~$x_a=1$ for all arcs~$a$ in~$H$ and~$y_a=1/m$ as well as $x_a=1$ for all arcs~$a$ incident to node~$t$.
    For all the remaining arcs in $G$ that are not in $H$, we set $y_a = x_a = 0$.
    To prove that $(x,y)$ is feasible, we show that the potential at node~$s$ does not exceed the value~$B$ if we send one unit of a potential-based flow from~$s$ to $t$ in the network induced by $y$.
    Let $(\pi,f)$ be the unit potential-based $s$--$t$ flow in the network induced by $y$ with zero potential at node $t$.
    We denote by $\pi_i\define \pi_{t_i}$ the potential at node $t_i$, and by $f_i\geq 0$ the flow from $t_i$ to $t$ for every $i\in [m]$.
    
    We claim that there exists at least one $j\in [m]$ such that $\pi_j\leq 1.$
    For the sake of a contradiction, assume that $\pi_i > 1$ for every $i$.
    This implies that the incoming flow at $t$ must be greater than $1$ because by the potential-flow equation together with the definition of $y$, we have
    \begin{align*}
        \sum_{i=1}^m f_i = \sum_{i=1}^m \frac{1}{m}{\pi}_i^{1/r} > 1,
    \end{align*}
    contradicting the flow conservation constraint at node $t$.
    Therefore, let $t_j$ be the terminal with potential $\pi_j\leq 1$ and let $P$ be an $s$--$t_j$ path in $H$.
    Every arc in $P$ has conductance $n\,m$ and can transport a flow of at most $1$. By the telescoping sum of the potential differences along each arc in $P$, we obtain
    \begin{align*}
        \pi_s -\pi_j = \sum_{a\in P}\frac{1}{n^rm^r}\abs{f_a}^r \leq \frac{n-1}{n^rm^r},
    \end{align*}
    which, together with $\pi_j \leq 1$, implies that $\pi_s\leq B$, so that $(x,y)$ is feasible.

    For the other direction, let $(x,y)$ be a feasible solution to the Network Design Problem constructed at the beginning of the proof, and let $H'\subseteq G'$ be the graph induced by $x$.
    Let $H= H'\cap G$ contain only nodes and arcs of $G$. 
    The graph $H$ contains  therefore all nodes and edges of $H'$ except the node $t$ and its incident edges, and has the same objective value as the solution $(x,y)$ by the definition of the fixed cost.
    We show that $H$ is a feasible solution to the Steiner tree problem.
    
    Without loss of generality, assume that $y_a=\bar{y}_a$ for all arcs with $x_a=1$ because there is no variable cost and increasing the conductance will only decrease the effective resistance, see \Cref{lem:effResMonotone}.
    Note that $H'$ must connect nodes $s$ and $t$, otherwise the effective resistance would be infinite.
    Furthermore, we can assume that $H'$ is connected; otherwise, removing the components of~$H'$ that do not contain $s$ and $t$ provides a feasible solution to the network design problem with possibly reduced cost.
    It is shown in~\cite{BornerKLPSS25} that for every $s$--$t$ cut $S$ in~$H'$ the sum of the conductances of the edges crossing this cut  must be at least $1/{B}^{1/r}$, or equivalently,
    \begin{equation} \label{eq:IPCO_cut}
        \frac{1}{\Big( \sum_{a\in \delta(S)} y_a \Big )^r} \leq {B} =1 + \frac{n-1}{n^rm^r} \quad \text{for all $s$--$t$ cuts $S$ in $H'$}.
    \end{equation}
    What remains to show is that the graph $H$ is connected and contains all terminal nodes. 
    If there exists more than one connected component in $H$, there must be one that contains at least one terminal node of $t_1,\dots, t_m$ but does not contain node $s$. 
    Otherwise, the graph~$H'$ would not be connected.
    Let $S\subseteq V$ be the component of $H$ not containing $s$ but at least one other terminal node.
    Then $S\cup\{t\}$ is an $s$--$t$ cut in $H'$ and if $S\neq \emptyset$.
    $\sum_{a\in \delta(S\cup\{t\})} y_a \leq \frac{m-1}{m}$
    Thus, for $r\geq 1$, we have
    \begin{align*}
        \frac{1}{\Big( \sum_{a\in \delta(S\cup\{t\})} y_a \Big )^r} & \geq  \left(\frac{m}{m-1}\right) ^r \\
        & = 1+ \frac{m^r-(m-1)^r}{(m-1)^r} \\
        & \geq 1+ \frac{1}{(m-1)^r} \\
        & > 1 + \frac{n-1}{n^rm^r},
    \end{align*}
    which contradicts \eqref{eq:IPCO_cut}.
\end{proof}

We next show that the Network Design Problem remains $\mathsf{NP}$-hard on series-parallel graphs, both when $\bar{y} \equiv \infty$ and when $\bar{y} \in \R_{\geq 0}^A$ and even when $r=1$.

\begin{restatable}{theorem}{NPhardnessUnbounded}
    \label{thm:NP_hardness}
    The Network Design Problem~\eqref{eq:problem} is $\mathsf{NP}$-hard on series-parallel graphs even for $r=1$ both when $\bar{y} \equiv \infty$ and when $\bar{y} \in \R_{\geq 0}^A$.
\end{restatable}

\begin{proof}
We first show hardness for the case that $\bar{y} \equiv \infty$.
In this case, the Network Design Problem reduces to Problem~\eqref{quadratic_path_problem} as shown in~\Cref{subsec:shortest-path}.
    We show that there exists a polynomial-time reduction from the decision variant of Problem~\eqref{quadratic_path_problem} to the Partition problem.
    Let $a_1,\dots,a_n\in\N$ be an instance of the partition problem and define $T\define \frac{1}{2}\sum_{i\in [n]}a_i$.
    Let $G$ be a graph on $n+1$ vertices $s= v_1, \dots, v_{n+1} =t$.
    For $i\in[n]$, the nodes $v_i$ and~$v_{i+1}$ are connected via two parallel edges $e_i^t$ and~$e_i^b$ which we call the top and bottom edge of bundle~$i$, respectively. 
    
    We set the parameters for Problem~\eqref{quadratic_path_problem} as follows
    \begin{align*}
        c'(e_i^t) &= a_i,            &&&     c'(e_i^b) &= 0, \\
        \gamma(e_i^t) &= 2T^{\frac{r+1}{r}}-T^{\frac{1}{r}}a_i,     &&&     \gamma(e_i^b) &= 2T^{\frac{r+1}{r}},
    \end{align*}
    so that all parameters are non-negative since $2T^{\frac{r+1}{r}}=2T\,T^{\frac{1}{r}} \geq a_i T^{\frac{1}{r}}$.
    We show that there exists a subset $S\subseteq [n]$ such that $\sum_{i\in S}a_i=T$ if and only if there exists an $s$--$t$~path~$P$ in $G$ such that the objective of that path with respect to $c'$ and $\gamma$ is at most $2nT^{\frac{r+1}{r}}$.
    
    Note that every $s$--$t$ path in $G$ corresponds to a partition of the numbers $a_1,\dots,a_n$ into two subsets, i.e., the top edges and the bottom edges in the path, and vice versa.
    More precisely, consider a subset $S\subseteq [n]$ and the corresponding path~$P$ such that the top edge of bundle $i$ is in $P$ if and only if $i\in S$ and the bottom edge exactly when $i\notin S$. 
    Then,
    \begin{align*}
            \Bigg(\sum_{a\in P} c'_a \Bigg)^{\!\!\frac{r+1}{r}} +\sum_{a\in P}\gamma_a =  \Bigg(\sum_{i\in S} a_i \Bigg)^{\!\!\frac{r+1}{r}} +2nT^{\frac{r+1}{r}}-T^{\frac{1}{r}}\sum_{i\in S} a_i \leq 2nT^{\frac{r+1}{r}}
    \end{align*}
    if and only if $\sum_{i\in S} a_i =T$.     

To see that the Network Design Problem is also $\mathsf{NP}$-hard in the presence of bounds $\bar{y} \in \R_{\geq 0}^A$ on the conductance, we note that when $c \equiv 0$, it is without loss of generality to assume that in every solution we have $y_a = \bar{y}_a$ for each arc~$a$ with $x_a = 1$. The result then follows from the $\mathsf{NP}$-hardness of Problem~\eqref{IPCO_st} shown in \Cref{thm:NP_hardness_bounded}.
\end{proof}


\bibliography{literature}

\end{document}